\documentclass[a4paper
               ,11pt
               ,oneside
               ,aps
               ,notitlepage
               ,nofootinbib
               ,superscriptaddress
               ,tightenlines
               ]
               {revtex4}

\usepackage{amsmath}
\usepackage{amssymb}  
\usepackage{amsthm}
\usepackage{amsfonts}
\usepackage{graphicx}
\usepackage{bbm}
\usepackage{url}
\usepackage{ upgreek }
\usepackage{dsfont}
\usepackage{color}

\newcommand{\ket}[1]{| #1 \rangle}

\newcommand{\bra}[1]{\langle #1 |}

\newcommand{\braket}[2]{\langle #1 | #2 \rangle}

\newcommand{\proj}[2]{| #1 \rangle\!\langle #2 |}

\newcommand{\id}{\ensuremath{\mathds{1}}}

\def\beq{\begin{equation}}
\def\eeq{\end{equation}}
\def\bq{\begin{quote}}
\def\eq{\end{quote}}
\def\ben{\begin{enumerate}}
\def\een{\end{enumerate}}
\def\bit{\begin{itemize}}
\def\eit{\end{itemize}}

\def\ra{\rightarrow}

\def\lb{\left(}
\def\rb{\right)}
\def\lset{\lbrace}
\def\rset{\rbrace}

\def\l|{\left|}
\def\r|{\right|}
\def\lbr{\left[}
\def\rbr{\right]}
\def\ident{\textnormal{id}}
\def\one{\id}

\newcommand\C{\mathbbm{C}}

\newcommand\R{\mathbbm{R}}
\newcommand\N{\mathbbm{N}}
\newcommand\M{\mathcal{M}}

\newcommand{\Lm}{\mathcal{L}}

\newcommand{\ketbra}[1]{|#1\rangle\langle#1|}

\theoremstyle{plain}
\newtheorem{thm}{Theorem}[section]
\newtheorem{lem}{Lemma}[section]
\newtheorem{cor}{Corollary}[section]
\newtheorem{prop}{Proposition}[section]
\newtheorem{defn}{Definition}[section]
\newtheorem{conj}{Conjecture}[section]
\theoremstyle{definition}

\begin{document}

\title[When Do Composed Maps Become Entanglement Breaking?]{When Do Composed Maps \\ Become Entanglement Breaking?}

\author{Matthias Christandl}
\email{christandl@math.ku.dk}
\affiliation{QMATH, Department of Mathematical Sciences, University of Copenhagen, Universitetsparken 5, 2100 Copenhagen, Denmark}

\author{Alexander M\"uller-Hermes}
\email{muellerh@posteo.net, muellerh@math.ku.dk}
\affiliation{QMATH, Department of Mathematical Sciences, University of Copenhagen, Universitetsparken 5, 2100 Copenhagen, Denmark}

\author{Michael M. Wolf}
\email{m.wolf@tum.de}
\affiliation{Zentrum Mathematik, Technische Universit\"{a}t M\"{u}nchen, 85748 Garching, Germany}

\date{\today}

\begin{abstract}

For many completely positive maps repeated compositions will eventually become entanglement breaking. To quantify this behaviour we develop a technique based on the Schmidt number: If a completely positive map breaks the entanglement with respect to any qubit ancilla, then applying it to part of a bipartite quantum state will result in a Schmidt number bounded away from the maximum possible value. Iterating this result puts a successively decreasing upper bound on the Schmidt number arising in this way from compositions of such a map. By applying this technique to completely positive maps in dimension three that are also completely copositive we prove the so called PPT squared conjecture in this dimension. We then give more examples of completely positive maps where our technique can be applied, e.g.~maps close to the completely depolarizing map, and maps of low rank. Finally, we study the PPT squared conjecture in more detail, establishing equivalent conjectures related to other parts of quantum information theory, and we prove the conjecture for Gaussian quantum channels. 

\end{abstract}

\maketitle

\vspace{-1.2cm}

\section{Introduction}

Let $\M_d$ denote the set of complex $d\times d$ matrices and let $\M_d^+$ denote the cone of positive matrices. We call a linear map $T:\M_{d_1}\ra \M_{d_2}$
\begin{itemize}
\item \textbf{positive} if $T(X)\in \M_{d_2}^+$ for any $X\in \M_{d_1}^+$.
\item \textbf{completely positive} if $\ident_n\otimes T : \M_n\otimes \M_{d_1}\ra \M_n\otimes \M_{d_2}$ is a positive map for all $n\in \N$, where $\ident_d$ denotes the identity map on $\M_d$.
\item \textbf{completely copositive} if $\vartheta_{d_2}\circ T$ is completely positive, where $\vartheta_d:\M_d\ra\M_d$ denotes the matrix transposition $\vartheta_d(X) = X^T$ in the computational basis of $\C^d$ (the definition does not depend on this choice).
\item \textbf{entanglement breaking} if for any positive matrix $X\in(\M_{d_2}\otimes \M_{d_1})^+$ the matrix $(\ident_{d_2}\otimes T)(X)$ is separable (i.e.~belongs to the convex cone generated by product matrices $A\otimes B$ with $A\in\M_{d_2}^+$ and $B\in\M_{d_1}^+$).
\end{itemize}
We will call a completely positive map a \emph{quantum channel} if it preserves the trace, i.e.~if $\text{Tr}\lbr T(X)\rbr=\text{Tr}\lbr X\rbr$ for any $X\in\M_{d_1}$. While quantum channels represent general physical processes in the context of quantum information theory, entanglement breaking quantum channels represent such processes that cannot be used to distribute entanglement, and they are useless for any non-classical communication task~\cite{horodecki2003entanglement}. Completely copositive quantum channels represent physical processes which are too noisy to be used for some information processing tasks (e.g.~quantum communication~\cite{holevo2001evaluating}). However, unlike entanglement breaking channels they can sometimes be used for certain non-classical information processing tasks (e.g.~private communication~\cite{horodecki2005secure}). 

\subsection{Motivation and previous results}

In this article we study when compositions of completely positive maps become entanglement breaking, which is a way of quantifying how ``noisy'' the physical process associated to the map is. There are various forms and related questions which can be (and have been) studied, mostly in the special case of quantum channels: Given a single quantum channel $T:\M_d\ra\M_d$ its \emph{entanglement breaking index} $n(T)$ has been defined in \cite{lami2015entanglement} as the smallest number $N\in\N$ such that all compositions of the form
\[
T\circ S_{N-1}\circ T\circ S_{N-2}\circ\cdots\circ S_1\circ T
\]
with quantum channels $S_i:\M_d\ra\M_d$ are entanglement breaking, or $n(T)=\infty$ if such a number does not exist. Related quantities have also been studied earlier in~\cite{de2012quantifying,de2013amendable}. In general it is not easy to compute the entanglement breaking index of a given quantum channel, and we refer to \cite{lami2015entanglement} for some examples where it is known. Related to the case of quantum channels $T:\M_d\ra\M_d$ with $n(T)=\infty$, it has been studied when \emph{none} of the compositions $T^n:=T\circ T\circ \cdots \circ T$ are entanglement breaking. In \cite{lami2016entanglement} such quantum channels were called \emph{entanglement saving}, and the subset of such quantum channels with full rank (as a linear map) has been characterized. Even stronger one may ask for which quantum channels $T:\M_d\ra\M_d$ \emph{none} of the limit points of the sequence $\lb T^n\rb_{n\in\N}$ are entanglement breaking (see \cite{lami2016entanglement} for details). Such quantum channels were called \emph{asymptotically entanglement saving}, and they have been characterized in \cite{lami2016entanglement} in terms of their fixed point algebra. It should be noted that the set of asymptotically entanglement saving channels is very small, e.g.~the unital quantum channels that are \emph{not} asymptotically entanglement saving are dense in the set of all unital quantum channels (see~\cite[Theorem 5.1.]{rahaman2018eventually}). 

We will be particularly interested in the case of completely positive maps that are also completely copositive (sometimes referred to as PPT maps). One of the authors has proposed the following conjecture~\cite{christandl2012PPT}.

\begin{conj}[PPT squared conjecture -- Version 1]
For any linear map $T:\M_d\ra\M_d$ that is completely positive and completely copositive its square $T\circ T$ is entanglement breaking.
\label{conj:PPT21}
\end{conj}

If true, this conjecture would imply limitations of using physical processes represented by completely positive and completely copositive maps in repeater scenarios (see~\cite{bauml2015limitations,christandl2017private} for more details), where the map effectively acts more than once. While the PPT squared conjecture is still open, there has been some recent progress on asymptotic versions of this conjecture: In \cite{kennedy2017composition} it has been shown that for a completely positive map $T:\M_d\ra\M_d$ that is also completely copositive and in addition unital\footnote{i.e.~mapping the identity matrix to itself} or trace-preserving, the limit points of the sequence $\lb T^n\rb_{n\in\N}$ are entanglement breaking. Moreover, it has been shown in \cite{rahaman2018eventually} that for any unital and completely copositive quantum channel $T:\M_d\ra\M_d$ there exists a finite $N\in\N$ such that $T^N$ is entanglement breaking. 

\subsection{Summary}

Contrary to most of the aforementioned results, we will not focus on asymptotic properties of compositions of completely positive maps. Instead we aim at results where given a map $T:\M_d\ra\M_d$ the composition $T^N$ is entanglement breaking for a finite $N\in\N$. Moreover, we will identify classes of completely positive maps, where this $N$ does not depend on the specific completely positive map, but only on the dimension $d$ of the input space. This will be a step towards statements similar to the PPT squared conjecture.  

Our article can be divided roughly into two parts on general techniques (Sections \ref{sec:SNTech}, \ref{sec:Examples}) and on the PPT squared conjecture (Sections \ref{sec:PPTsquareEquiv}, \ref{sec:ExPPT2}). In the first part we introduce new techniques for studying questions of the above kind based on the Schmidt number~\cite{terhal2000schmidt}, an integer valued entanglement measure for positive bipartite matrices. In Section \ref{sec:SNTech} we show that any completely positive map $T:\M_d\ra\M_d$ breaking the entanglement between its input and any $2$-dimensional reference system (called $2$-entanglement breaking in the following) will become fully entanglement breaking after $d-1$ compositions. Unfortunately, these techniques do not apply to linear maps that are completely positive and completely copositive without assumptions on the dimension. Assuming a conjectured bound on the Schmidt number of states with positive partial transpose we introduce a different iteration technique that would apply to general completely positive maps that are also completely copositive. In Section \ref{sec:Examples} we prove the PPT squared conjecture for $d=3$, and discuss some examples of $2$-entanglement breaking maps (that are not entanglement breaking) in higher dimensions, e.g.~positive maps close enough to the completely depolarizing map, and $2$-positive maps of low rank. The second part of our article is on the PPT squared conjecture itself. In Section \ref{sec:PPTsquareEquiv} we state three equivalent formulations of the PPT squared conjecture linking it to other parts of quantum information theory and questions related to positive maps between matrix algebras. Finally, in Section \ref{sec:ExPPT2} we present some examples where the PPT squared conjecture holds true, and most importantly we prove the conjecture for Gaussian channels.

\subsection{Notation and preliminaries}

We will denote by $\mathcal{U}_d\subset \M_d$ the subset of unitary $d\times d$ matrices. Throughout this article we will denote the computational basis by $\ket{i}\in \C^d$, i.e. the vector with a single $1$ in the $i$th component and zeros in the other ones. A natural basis of $\M_d$ is then given by the matrix units $\proj{i}{j}\in\M_d$ having a single $1$ in their $(i,j)$ entry. We will denote the (unnormalized) $d$-dimensional maximally entangled state by $\ket{\Omega_d} = \sum^d_{i=1}(\ket{i}\otimes \ket{i})\in \C^d\otimes \C^d$ and the corresponding matrix by $\omega_d := \proj{\Omega_d}{\Omega_d}$. The $d$-dimensional identity matrix will be denoted by $\one_d$. For a linear transformation $Y:\C^{d_1}\ra\C^{d_2}$ we will denote by $\text{Ad}_Y:\M_{d_1}\ra\M_{d_2}$ the completely positive map defined as $\text{Ad}_Y(X) = YXY^\dagger$ for $X\in\M_d$. 

Given a linear map $L:\M_{d_1}\ra\M_{d_2}$ we call the matrix 
\[
C_L :=(\ident_{d_1}\otimes L)(\omega_{d_1})\in\M_{d_1}\otimes \M_{d_2}
\] 
the Choi matrix associated to $L$ (see~\cite{choi1975completely}). The map $L\mapsto C_L$ defines an isomorphism, called the Choi-Jamiolkowski isomorphism, between linear maps $L:\M_{d_1}\ra\M_{d_2}$ and matrices in $\M_{d_1}\otimes \M_{d_2}$. Completely positive maps correspond to positive matrices under this isomorphism~\cite{choi1975completely}, and occasionally we will define completely positive maps from positive matrices in this way. 

The following lemma collects some frequently used and well-known techniques involving the maximally entangled state and linear maps that can be proven by direct computation.

\begin{lem}[Maximally entangled state]\hfill
\begin{enumerate}
\item For any $d_2\times d_1$-matrix $X$ we have $\lb \id_{d_1}\otimes X\rb\ket{\Omega_{d_1}} = \lb X^T\otimes \id_{d_2} \rb\ket{\Omega_{d_2}}$.
\item For any map $\Lm:\M_{d_1}\ra\M_{d_2}$ that is Hermiticity-preserving (i.e. maps Hermitian matrices to Hermitian matrices), we have $\lb\ident_{d_1}\otimes \Lm\rb\lb \omega_{d_1}\rb = \lb \vartheta_{d_1}\circ\Lm^*\circ\vartheta_{d_2}\otimes \ident_{d_2}\rb\lb \omega_{d_2}\rb$. 
\end{enumerate}
In the above $\Lm^*$ denotes the adjoint with respect to the Hilbert-Schmidt inner product.
\label{Lemma:tricks}
\end{lem}

The following definition is well-known:

\begin{defn}[Schmidt rank]
The Schmidt-rank of a bipartite vector $\ket{\psi}\in \C^{d_1}\otimes \C^{d_2}$ is defined as 
\[
\text{SR}\lb\ket{\psi}\rb := \text{rk}\lb (\ident_{d_1}\otimes \text{tr})\lb \proj{\psi}{\psi}\rb\rb.
\]
\label{defn:SchmidtR}
\end{defn}

The following lemma is a version of the well-known Schmidt decomposition (equivalent to the singular value decomposition):

\begin{lem}[Schmidt decomposition]
Any vector $\ket{\psi}\in \C^{d_1}\otimes \C^{d_2}$ with Schmidt-rank $k=\text{SR}\lb\ket{\psi}\rb$ can be decomposed as 
\[
\ket{\psi} = (\one_{d_1}\otimes V)(\ket{\psi'}) = (W\otimes \one_{d_2})(\ket{\psi''})
\]
with isometries $V:\C^{k}\ra \C^{d_2}$ and $W:\C^{k}\ra \C^{d_1}$ and vectors $\ket{\psi'}\in \C^{d_1}\otimes \C^{k}$ and $\ket{\psi''}\in \C^{k}\otimes \C^{d_2}$.
\label{lem:SchmidtDecomp}
\end{lem}

The following lemma can be obtained easily from the Schmidt decomposition: 

\begin{lem}
Any vector $\ket{\psi}\in \C^{d_1}\otimes \C^{d_2}$ can be written as 
\[ 
\ket{\psi} = (\one_{d_1}\otimes A)(\ket{\Omega_{d_1}}) = (B\otimes \one_{d_2})(\ket{\Omega_{d_2}})
\]
with linear maps $A:\C^{d_1}\ra \C^{d_2}$ and $B:\C^{d_2}\ra \C^{d_1}$. 
\label{lem:YAVD}
\end{lem}

\section{Schmidt number techniques}
\label{sec:SNTech}

We will start with the following generalization of the Schmidt rank to arbitrary positive matrices first studied in \cite{terhal2000schmidt}. 

\begin{defn}[Schmidt number]
The Schmidt number of a bipartite positive matrix $X\in \M(\C^{d_1}\otimes \C^{d_2})^+$ is the minimal number $k\in \N$ such that 
\[
X = \sum^l_{i=1} \proj{\psi_i}{\psi_i}
\] 
for some $l\in \N$ and vectors $\ket{\psi_i}\in \C^{d_1}\otimes \C^{d_2}$ with $\text{SR}\lb\ket{\psi_i}\rb\leq k$. 
\end{defn}

We denote the Schmidt number of a matrix $X\in (\M_{d_1}\otimes \M_{d_2})^+$ by $\text{SN}(X)$. Note that $\text{SN}(X)\in \lset 1,\ldots , \min(d_1,d_2)\rset$ with $\text{SN}(X)=1$ if and only if $X$ is separable.

\subsection{Schmidt number iteration}

To introduce the Schmidt number iteration technique we need to define a certain class of linear maps. To make our statements as general as possible, we will formulate them in terms of $k$-positive maps, i.e.~linear maps $T:\M_{d_1}\ra\M_{d_2}$ such that $\ident_k\otimes T:\M_k\otimes \M_{d_1}\ra\M_k\otimes \M_{d_2}$ is positive. However, for most applications we only consider completely positive maps (i.e.~$\min(d_1,d_2)$-positive maps).   

\begin{defn}[$n$-entanglement breaking maps]
A $k$-positive map $T:\M_{d_1}\ra\M_{d_2}$ is called $n$-entanglement breaking for some $n\leq k$ when
\[
(\ident_n\otimes T)(X)\text{ is separable }
\]
for any $X\in(\M_n\otimes \M_{d_1})^+$. 
\label{defn:kEB}
\end{defn}

Note that for $n\geq d_1$ any $n$-entanglement breaking map is entanglement breaking in the usual sense~\cite{horodecki2003entanglement}, and every $n$-entanglement breaking map is also $n'$-entanglement breaking for any $n'\leq n$. The relevance of $n$-entanglement breaking maps lies in the following lemma that is central for our technique.

\begin{lem}[Schmidt number trimming]
A $k$-positive map $T:\M_{d_1}\ra\M_{d_2}$ is $n$-entanglement breaking for $n\leq \min(k,d_2)$ if and only if for any $l\leq \min(k,d_2)$ we have
\begin{equation}
\text{SN}\lb (\ident_{l}\otimes T)(X)\rb \leq \max(l-n+1,1)
\label{equ:SNTrim}
\end{equation}
for any $X\in(\M_l\otimes \M_{d_1})^+$.
\label{lem:SNtrimmingFrom2EB}
\end{lem}
\begin{proof}
If $T:\M_{d_1}\ra\M_{d_2}$ satisfies \eqref{equ:SNTrim} for $l=n$ it is in particular $n$-entanglement breaking.      

For the converse direction note first that the statement is obvious for $l\leq n$. Now, assume that for $n< l\leq \min(k,d_2)$ there exists an $X\in(\M_{l}\otimes \M_{d_1})^+$ such that 
\[
\text{SN}\lb(\ident_l\otimes T)(X)\rb \geq l-n+2.
\]
We can write 
\[
(\ident_l\otimes T)(X) = \sum^l_{i,j=1} \proj{i}{j}\otimes T(X_{ij})\in \lb \M_l\otimes \M_{d_2}\rb^+.
\]
Since $l\leq d_2$ by Theorem \ref{thm:SubBlockEntSN} there exist $\lset i_1,i_2,\ldots , i_n\rset\subset\lset 1,\ldots ,l\rset$ such that the matrix
\[
Y = \sum^n_{s,t=1} \proj{s}{t}\otimes T(X_{i_si_t}) \in \lb \M_n\otimes \M_{d_2}\rb^+
\]
is entangled. Since $T$ is $n$-entanglement breaking, this gives a contradiction.
\end{proof}

It should also be noted that the our definition of $n$-entanglement breaking maps differs from the $n$-partially entanglement breaking maps (also known as $n$-superpositive maps~\cite{skowronek2009cones}) introduced by Chru{\'s}ci{\'n}ski and Kossakowski in \cite{chruscinski2006partially}. Specifically, we can apply Lemma \ref{lem:SNtrimmingFrom2EB} to the Choi matrix of an $n$-entanglement breaking, completely positive map $T:\M_{d_1}\ra\M_{d_2}$ with $\max(d_1,n)\leq d_2$ showing that it is $\max(d_1-n+1,1)$-partially entanglement breaking according to~\cite{chruscinski2006partially}. However, the converse is not true: The completely positive map $\text{Ad}_A:\M_d\ra\M_d$ with $A = \one_{d-1}\oplus 0$ is $(d-1)$-partially entanglement breaking, but not $2$-entanglement breaking according to our definition for any $d>2$. 

The amount of entanglement quantified by the Schmidt number, that is still present after applying an $n$-entanglement breaking map to part of a positive matrix (with dimensions satisfying the assumptions of Lemma \ref{lem:SNtrimmingFrom2EB}), is bounded away from the maximal possible value by $n$. Using a simple observation about the connection between the Schmidt number and the effective dimension of a state leads to the following Schmidt number iteration technique:

\begin{thm}[Schmidt number iteration]
If for each $i\in \lset 1,\ldots ,\lceil \frac{d-1}{n-1}\rceil\rset$ the completely positive maps $T_i:\M_d\ra\M_d$ are $n$-entanglement breaking, then the composition $T_{\lceil \frac{d-1}{n-1}\rceil}\circ \cdots \circ T_1$ is entanglement breaking.
\label{thm:SNIteration}
\end{thm}
\begin{proof}
Consider $k\in \lset 2,\ldots ,d\rset$ and $X\in(\M_d\otimes \M_d)^+$ with $\text{SN}\lb X\rb=k$. By the Schmidt decomposition (see Lemma \ref{lem:SchmidtDecomp}) there exists an $l\in \N$, (unnormalized) pure states $\psi_j\in (\M_{k}\otimes \M_{d})^+$ and isometries $V_j:\C^k\ra \C^d$ for every $j\in \lset 1,\ldots , l\rset$ such that 
\[
X = \sum^l_{j=1}(V_j\otimes \one_d)\psi_j(V^\dagger_j\otimes \one_d).
\]
Now for any $n$-entanglement breaking map $T:\M_d\ra\M_d$ (see Definition \ref{defn:kEB}) we have 
\[
\text{SN}\lb(\ident_k\otimes T)(\psi_j)\rb \leq \max(k-n+1,1)
\]
for all $j\in \lset 1,\ldots , l\rset$. This shows that 
\begin{align*}
\text{SN}\lb(\ident_d\otimes T)(X)\rb &= \text{SN}\lb\sum^l_{j=1} (V_j\otimes \one_d)(\ident_k\otimes T)\lb\psi_j\rb(V^\dagger_j\otimes \one_d)\rb \\
& \leq \max_j ~\text{SN}\lb (V_j\otimes \one_d)(\ident_k\otimes T)\lb\psi_j\rb(V^\dagger_j\otimes \one_d)\rb\\
&\leq \max(k-n+1,1).
\end{align*}
Applying the above argument $\lceil \frac{d-1}{n-1}\rceil$-times for the successive application of the maps $T_i$ for $i\in \lset 1,\ldots ,\lceil \frac{d-1}{n-1}\rceil\rset$ shows that for any $X\in (\M_{d}\otimes \M_{d})^+$ we have 
\[
\text{SN}\lb(\ident_d\otimes T_{\lceil \frac{d-1}{n-1}\rceil}\circ \cdots \circ T_1)(X) \rb = 1,
\]
which finishes the proof.

\end{proof}

\subsection{Alternative iteration from a Schmidt number conjecture}
\label{sec:AltIt}

The techniques from the previous section only apply to $k$-entanglement breaking maps. We will later give some examples of such maps (see Section \ref{sec:Examples}), but it should be noted that not all completely positive and completely copositive maps in arbitrary dimensions belong to this class. Consider for example a completely positive map $T:\M_2\ra\M_4$ that is completely copositive, but not entanglement breaking (e.g.~a map corresponding to the Tang-Horodecki state~\cite{tang1986positive,horodecki1997separability} via the Choi-Jamiolkowski isomorphism). Then, we can consider the linear map $\tilde{T}:\M_4\ra\M_4$ given by 
\[
\tilde{T}(X) = T\lb(\one_2\otimes \bra{0}) X(\one_2\otimes \ket{0})\rb.
\]
It is easy to see that $\tilde{T}$ is completely positive and completely copositive, but \emph{not} $2$-entanglement breaking. 

Since the techniques from the previous section do not apply immediately to linear maps that are both completely positive and completely copositive, we will present here an alternative iteration technique. Unfortunately, this technique still relies on the following conjecture, first stated by Sanpera, Bru{\ss} and Lewenstein in~\cite{sanpera2001schmidt}. 

\begin{conj}
If $T:\M_{d}\ra\M_{d}$ is completely positive and completely copositive, then
\[
\text{SN}\lb C_T\rb\leq d-1.
\] 
\label{conj:MaxSNPPT}
\end{conj}

It is well-known that the previous conjecture is true for $d=2$ (see \cite{woronowicz1976positive}) and recently it has been established for $d=3$ (see \cite{yang2016all}). Note that the same conjecture with different input and output dimensions is false due to the existence of the entangled Tang-Horodecki~\cite{tang1986positive,horodecki1997separability} state having positive partial transpose.

Unfortunately, Conjecture \ref{conj:MaxSNPPT} is still unsolved for general dimensions $d\in\N$. If it were true, we could use the following iteration technique. We will start with a lemma.

\begin{lem}
If Conjecture \ref{conj:MaxSNPPT} is true, then for any $S_1,S_2:\M_d\ra\M_d$ satisfying 
\[
\text{SN}(C_{S_1})\leq k \hspace*{0.5cm}\text{ and }\hspace{0.5cm}\text{SN}(C_{S_2})\leq k
\]
and any $T:\M_d\ra\M_d$ that is completely positive and completely copositive, we have
\[
\text{SN}(C_{S_1\circ T\circ S_2})\leq k-1.
\]
\label{lem:OtherIteration}
\end{lem}

\begin{proof}
By \cite[Theorem 1 and (8)]{chruscinski2006partially} there exist $l\in\N$ and operators $A_i,B_i\in\M_d$ satisfying $\text{rk}(A_i)\leq k$ and $\text{rk}(B_i)\leq k$ for any $i\in \lset 1,\ldots ,l\rset$ such that 
\[
S_1(X) = \sum^l_{i=1}A_iXA^\dagger_i \hspace*{0.5cm}\text{ and }\hspace{0.5cm} S_2(X) = \sum^l_{i=1}B_iXB^\dagger_i
\]
for any $X\in\M_d$. For any $i\in\lset 1,\ldots l\rset$ we can use the singular value decomposition to decompose $A_i = U_i D_i V^\dagger_i$ and $B_i = \tilde{U}_i \tilde{D}_i \tilde{V}^\dagger_i $ with isometries $U_i,V_i,\tilde{U}_i,\tilde{V}_i:\C^{k}\ra \C^d$ and positive diagonal matrices $D_i,\tilde{D}_i\in \M_k$.  

Then we have 
\[
S_1\circ T\circ S_2 = \sum^l_{i,j=1} \text{Ad}_{U_i}\circ \text{Ad}_{D_i}\circ \text{Ad}_{V_i}^*\circ T\circ \text{Ad}_{\tilde{U}_j}\circ \text{Ad}_{\tilde{D}_j}\circ \text{Ad}_{\tilde{V}_j}^*.
\]
For any $i,j\in \lset 1,\ldots , l\rset$ the linear map 
\[
K_{ij} := \text{Ad}_{D_i}\circ \text{Ad}_{V_i}^*\circ T\circ \text{Ad}_{\tilde{U}_j}\circ \text{Ad}_{\tilde{D}_j}:\M_k\ra\M_k
\] 
is completely positive and completely copositive and if Conjecture \ref{conj:MaxSNPPT} is true, then we would have that $\text{SN}(C_{K_{ij}})\leq k-1$. Since the Schmidt number cannot increase under separable operations and under forming sums we find that 
\[
\text{SN}\lb C_{S_1\circ T\circ S_2}\rb = \text{SN}\lb \sum^l_{i,j=1} ((\vartheta_{d}\circ \text{Ad}_{\tilde{V}_j}\circ \vartheta_{k})\otimes \text{Ad}_{U_i})\lb C_{K_{ij}}\rb\rb\leq k-1.
\] 
 
\end{proof}

\begin{thm}
If Conjecture \ref{conj:MaxSNPPT} is true, we have for any completely positive and completely copositive map $T:\M_{d}\ra\M_{d}$ that 
\[
\text{SN}(C_{T^{2^{k}-1}})\leq d-k.
\]  
In particular the composition $T^{2^{d-1}-1}$ would be entanglement breaking.  
\end{thm}

\begin{proof}
We will use induction for the proof. The case $k=1$ follows directly from Conjecture \ref{conj:MaxSNPPT}. Now suppose that the theorem holds for $k\in\lset 1,\ldots ,d-2\rset$ inserting $S_1,S_2=T^{2^{k}-1}$ in Lemma \ref{lem:OtherIteration} we obtain 
\[
\text{SN}(C_{T^{2^{k+1}-1}})=\text{SN}(C_{T^{2^{k}-1}\circ T\circ T^{2^{k}-1}})\leq d-k-1.
\]

\end{proof}

\section{Examples of $k$-entanglement breaking maps}
\label{sec:Examples}

In this section we will study examples of completely positive maps that are $k$-entanglement breaking for some $k\in\N$ without being entanglement breaking.

\subsection{Examples for dimension $d=3$}
\label{subsec:d3}

By~\cite{woronowicz1976positive} a positive matrix $X\in (\M_2\otimes \M_3)^+$ is separable if and only if it has positive partial transpose. This immediately characterizes the set of $2$-entanglement breaking maps. 

\begin{thm} 
A linear map $T:\M_{3}\ra \M_3$ is $2$-entanglement breaking if and only if it is $2$-positive and $2$-copositive (i.e.~$\vartheta_3\circ T$ is $2$-positive).
\label{thm:charact2EBd3}
\end{thm}

Since any completely positive map is in particular $2$-positive, we find that any completely positive and completely copositive map $T:\M_{3}\ra \M_3$ is $2$-entanglement breaking. Now an application of Lemma \ref{lem:SNtrimmingFrom2EB} and Theorem \ref{thm:SNIteration} shows that the PPT squared conjecture holds in dimension $d=3$:

\begin{cor}
For any pair of completely positive and completely copositive maps $T_1, T_2:\M_{3}\ra\M_{3}$ the composition $T_2\circ T_1$ is entanglement breaking.
\label{thm:PPT2Qutrit}
\end{cor}

We will see in Section \ref{sec:PPTsquareEquiv} that the PPT squared conjecture implies Conjecture \ref{conj:PPT23} on entanglement annihilation. This easily gives the following corollary.

\begin{cor}
For any pair of linear maps $T_1, T_2:\M_{3}\ra\M_{3}$ both of which completely positive and completely copositive the matrix $(T_1\otimes T_2)(X)$ is separable for any positive matrix $X\in \M(\C^{3}\otimes \C^{3})^+$. 
\end{cor}

\subsection{Examples for dimension $d=4$}
\label{subsec:d4} 

In \cite{huber2018high} it was shown (also see Appendix \ref{sec:CWDEC}) that for $d_1\leq d_2$ any bipartite positive matrix $X\in(\M_{d_1}\otimes \M_{d_2})^+$ satisfying $(\vartheta_{d_1}\otimes \ident_{d_2})(X) = X$, i.e.~it is invariant under partial transposition of the smaller subsystem, has Schmidt number $\text{SN}(X)\leq d_1-1$. The following lemma is a straightforward consequence of these results in combination with a theorem from \cite{chruscinski2006partially}.

\begin{lem}
For any $T:\M_{d_1}\ra \M_{d_2}$ satisfying $T\circ \vartheta_{d_1} = T$ in the case $d_1\leq d_2$, and $\vartheta_{d_2}\circ T= T$ in the case $d_2\leq d_1$, there exists $k\in \N$ and matrices $A_i:\C^{d_1}\ra \C^{d_2}$ with $\text{rk}(A_i)\leq \min(d_1,d_2) -1$ such that
\[
T(X) = \sum^k_{i=1} A_i XA^\dagger_i
\]
for any $X\in \M_{d_1}$.
\label{lem:KrausRankBound}
\end{lem}
\begin{proof}
Assuming $d_1\leq d_2$ and $T\circ \vartheta_{d_1} = T$, it can be easily verified that the Choi matrix $C_T = (\ident_{d_1}\otimes T)(\omega_{d_1})$ satisfies the assumptions of Theorem \ref{thm:SNPTInv} from Appendix \ref{sec:CWDEC}, and $\text{SN}(C_T)\leq d_1-1$. Now an application of \cite[Theorem 1 and (8)]{chruscinski2006partially} finishes the proof. The case $d_2\leq d_1$ and $\vartheta_{d_2}\circ T= T$ works in the same way. 
\end{proof}

Note that in the case $d_1=d_2=d$ we can apply the previous lemma when $T\circ \vartheta_{d} = T$ or $\vartheta_d \circ T= T$. Note that these cases are \emph{not} equivalent. In the special case of $d_1=d_2=4$ we have:

\begin{thm}
Let $T:\M_{4}\ra \M_{4}$ be completely positive and completely copositive, and let $S:\M_{4}\ra \M_{4}$ be completely positive such that $\vartheta_{4}\circ S = S$ or $S\circ \vartheta_{4} = S$. Then the composition $S\circ T$ is $2$-entanglement breaking.
\label{thm:PTInv2EB}
\end{thm}
\begin{proof}
By Lemma \ref{lem:KrausRankBound} there exists $k\in \N$ and matrices $A_i\in\M_4$ with $\text{rk}(A_i)\leq 3$ such that
\[
S(X) = \sum^k_{i=1} A_i XA^\dagger_i
\]
for any $X\in \M_{4}$. Consider the singular value decomposition $A_i = U_iD_iV_i$ with unitaries $U_i,V_i\in\mathcal{U}_4$ and $D_i = \text{diag}(d^1_i,d^2_i,d^3_i,0)$ with $d^1_i,d^2_i,d^3_i\in\R_0^+$. Note that for any $i\in\lset 1,\ldots ,k\rset$ we have 
\[
D_iV_i T(X) V^\dagger_i D_i = K_i(X)\oplus 0
\]
for some completely positive and completely copositive $K_i:\M_4\ra \M_3$. For any $Y\in(\M_2\otimes \M_4)^+$ the matrix $(\ident_2\otimes K_i)(Y)\in(\M_2\otimes \M_3)^+$ has a positive partial transpose and thus is separable by \cite{woronowicz1976positive}. Since separability is preserved under local unitary transformations we have that  
\[
(\ident_2\otimes \text{Ad}_{A_i}\circ T)(Y)= (\one_2\otimes U_i)\lbr (\ident_2\otimes K_i)(Y)\oplus 0\rbr(\one_2\otimes U_i)^\dagger
\]
is separable as well. Using that sums of separable operators are separable we find that $(\ident_2\otimes S\circ T)(Y)$ is separable. This finishes the proof.

\end{proof}

Using the Schmidt number iteration we obtain the following corollary.

\begin{cor}
For $i\in \lset 1,2,3\rset$ let $T_i:\M_{4}\ra\M_{4}$ denote completely positive and completely copositive maps and let $S_i:\M_{4}\ra\M_4$ denote completely positive maps satisfying $\vartheta_{4}\circ S_i = S_i$ or $S_i\circ \vartheta_{4} = S_i$. Then the composition $S_3\circ T_3\circ S_2\circ T_2\circ S_1\circ T_1$ is entanglement breaking.

\end{cor}

\begin{proof}
For all $i\in\lset 1,2,3\rset$ the map $S_i\circ T_i$ is $2$-entanglement breaking by Theorem \ref{thm:PTInv2EB}. Now applying the Schmidt number iteration technique from Theorem \ref{thm:SNIteration} finishes the proof.
\end{proof}

It should be noted that Theorem \ref{thm:PTInv2EB} holds more generally with the same proof for completely positive $S:\M_{4}\ra \M_{4}$ that are $3$-partially entanglement breaking (see~\cite{chruscinski2006partially}), i.e.~such that the Choi matrix satisfies $\text{SN}\lb C_S\rb\leq 3$. This would include all linear maps that are both completely positive maps and completely copositive if Conjecture \ref{conj:MaxSNPPT} were true in dimension $d=4$. 

\subsection{Examples close to the completely depolarizing map}

In this section we will characterize a class of $2$-entanglement breaking maps by distance to the completely depolarizing map $X\mapsto \text{Tr}\lbr X\rbr\one_d$. Our result has some similarity to the results of Gurvits and Barnum (see~\cite{gurvits2002largest}) on balls containing only separable states around the maximally mixed state, and it should be noted that our proof relies on this in an indirect way.  

\begin{thm}
Any positive map $P:\M_d\ra \M_d$ satisfying
\[
\| P(X) - \text{Tr}\lbr X\rbr\one_d\|_\infty\leq \frac{1}{2}\| X\|_\infty,
\]
for any $X\in\M_d$, is $2$-entanglement breaking, and in particular $2$-positive.
\end{thm}

\begin{proof}
Any positive matrix $Z\in(\M_2\otimes \M_d)^+$ can be written in block-form as
\[
Z = \begin{pmatrix} A & B\\ B^\dagger & C\end{pmatrix},
\] 
with $A,B,C\in\M_d$. For showing that $(\ident_2\otimes P)(Z)$ is separable, we can assume without loss of generality that the marginal $M=(\ident_2\otimes \text{Tr})(Z)\in \M^+_2$ is invertible (we can otherwise consider $Z+\epsilon\one_2\otimes \one_d$ for $\epsilon\ra 0$). We can then define the positive matrix 
\begin{equation}
\begin{pmatrix} \rho & X\\ X^\dagger & \sigma\end{pmatrix} = (M^{-1/2}\otimes \one_d)Z(M^{-1/2}\otimes \one_d),
\label{equ:normalizedMat}
\end{equation}
for $\rho,\sigma\geq 0$ satisfying $\text{Tr}\lbr \rho\rbr =\text{Tr}\lbr \sigma\rbr=1$, and $\text{Tr}\lbr X\rbr=0$. Clearly, $(\ident_2\otimes P)(Z)$ is separable if and only if 
\[
(M^{-1/2}\otimes \one_d)(\ident_2\otimes P)(Z)(M^{-1/2}\otimes \one_d) = \begin{pmatrix} P(\rho) &  P(X)\\ P(X)^\dagger & P(\sigma)\end{pmatrix}
\]
is separable. By \cite[Lemma 1]{johnston2013separability} it is sufficient to show that 
\begin{equation}
\| P(X)\|^2_\infty\leq \lambda_{\min}\lbr P(\rho)\rbr\lambda_{\min}\lbr P(\sigma)\rbr,
\label{equ:Goal}
\end{equation}
where $\lambda_{\min}\lbr\cdot\rbr$ denotes the minimal eigenvalue. For this, note that positivity of \eqref{equ:normalizedMat} implies $\text{Ran}\lb X^\dagger\rb\subset \text{Ran}\lb \sigma\rb = \text{supp}\lb \sigma\rb$ and
\[
\rho \geq X \sigma^{-1} X^\dagger,
\]
where $\sigma^{-1}$ denotes the Moore-Penrose pseudoinverse of $\sigma$. Using that $\sigma^{-1}\geq \frac{1}{\|\sigma\|_{\infty}} Q$ for a projection $Q$ onto the support of $\sigma$, we obtain
\[
\|\rho\|_\infty \|\sigma\|_\infty\geq \|X\|^2_\infty.
\] 
Furthermore, since $\rho\geq 0$ with $\text{Tr}\lbr \rho\rbr =1$ the assumptions of the theorem imply
\[
1-\lambda_{\min}\lbr P(\rho)\rbr\leq \|P(\rho) - \one_d \|_\infty \leq \frac{1}{2}\|\rho\|_\infty
\] 
and the same holds for $\sigma$. Combining the previous estimates, and using the assumptions of the theorem again with $\text{Tr}\lbr X\rbr = 0$ implies
\begin{align*}
\|P(X)\|^2_\infty &\leq \frac{1}{4}\|X\|^2_\infty \leq \frac{1}{4}\|\rho\|_\infty \|\sigma\|_\infty \\
&\leq \lb 1- \frac{1}{2}\|\rho\|_\infty\rb\lb 1- \frac{1}{2}\|\sigma\|_\infty\rb\leq \lambda_{\min}\lbr P(\rho)\rbr\lambda_{\min}\lbr P(\sigma)\rbr.
\end{align*}
Here the third inequality follows from $\rho,\sigma\geq 0$ with $\text{Tr}\lbr \rho\rbr =\text{Tr}\lbr \sigma\rbr=1$. This proves \eqref{equ:Goal} and finishes the proof.

\end{proof}

The previous theorem can be used to show that completely positive maps can be $2$-entanglement breaking without being completely copositive: For $p\in\lbr -1,1\rbr$ we denote by $W_p:\M_d\ra\M_d$ the (unnormalized) Holevo-Werner map given by 
\begin{equation}
W_p(X) = \text{Tr}\lbr X\rbr\one_d - pX^T
\label{equ:HWMaps}
\end{equation}
for any $X\in\M_d$. It is well-known that $W_p$ is entanglement breaking if and only if it is completely copositive which is the case for $p\in\lbr 1/d, 1\rbr$. It is clear that Holevo-Werner map $W_p$ cannot be $2$-entanglement breaking for $p>1/2$, because its restriction to $\M_2\oplus 0_{d-2}$ (containing $W_{p}:\M_2\ra\M_2$) is not entanglement breaking. However, using the previous theorem we easily obtain:

\begin{cor}
For $d\in\N$ the Holevo-Werner map $W_{p}:\M_d\ra\M_d$ is $2$-entanglement breaking if and only if $p\in\lbr -1,1/2\rbr$.
\end{cor}

While any entanglement breaking map is necessarily completely copositive and completely positive, the previous corollary shows that some completely positive $2$-entanglement breaking maps are not completely copositive. It should be noted that applying the Schmidt number iteration to Holevo-Werner maps is not interesting, since it is easily verified that $W_{p}\circ W_p$ is already entanglement breaking for any $p\in\lbr -1,1\rbr$. Moreover, using \cite[Theorem 7]{lami2016bipartite} it can be shown that for any $p\in\lbr-1, \sqrt{3}-1\rbr$ the output $(W_{p}\otimes W_{p})\lb X\rb$ is separable for any positive matrix $X\in (\M_d\otimes \M_d)^+$. This shows that for such $p$ even the map $W_p\circ S\circ W_p$ is entanglement breaking for any completely positive map $S:\M_d\ra\M_d$.

\subsection{$2$-entanglement breaking maps from rank}

In this section we will study how $2$-positive maps with low rank (as linear maps) are $2$-entanglement breaking. This will be based on the surprising separability criteria by Cariello~\cite{cariello2014separability,cariello2015does} for positive matrices of low operator Schmidt rank. We will start by defining these terms:

\begin{defn}[Operator Schmidt rank]
Given a bipartite matrix $Y\in \M_{d_1}\otimes \M_{d_2}$ we define its operator Schmidt rank as the unique number $\mathcal{R}(Y)\in\N$ such that 
\[
Y = \sum^{\mathcal{R}\lb Y\rb}_{i=1} A_i\otimes B_i,
\]
with sets of non-zero, mutually orthogonal\footnote{with respect to the Hilbert-Schmidt inner product.} operators $\lset A_i\rset^{\mathcal{R}\lb Y\rb}_{i=1}\subset \M_{d_1}$ and $\lset B_i\rset^{\mathcal{R}\lb Y\rb}_{i=1}\subset\M_{d_2}$. 
\label{defn:OSRank}
\end{defn}

It turns out that the operator Schmidt rank of a Choi matrix coincides with the rank of the linear map it corresponds to. For completeness we prove this fact in Appendix \ref{sec:Appendix2}. Note that for Hermitian matrices $Y$ the operators $\lset A_i\rset_i$ and $\lset B_i\rset_i$ in the operator Schmidt decomposition can be chosen Hermitian as well (see~\cite[Theorem 1.30]{cariello2014separability}). With this we can show:

\begin{thm}
Any $2$-positive map $P:\M_{d_1}\ra\M_{d_2}$ with $\text{rk}\lb P\rb\leq 3$ is $2$-entanglement breaking.
\label{thm:2EBfromRank}
\end{thm}

\begin{proof}
Given a positive matrix $X\in(\M_{2}\otimes \M_{d_1})^+$ we can use the Choi-Jamiolkowski isomorphism to find a completely positive map $T:\M_2\ra\M_{d_1}$ such that $X = C_T$. Now, we have 
\[
(\ident_2\otimes P)(X) = C_{P\circ T}\geq 0,
\]
since $P$ is $2$-positive. As $\text{rk}\lb P\rb\leq 3$ we also have $\text{rk}\lb P\circ T\rb\leq 3$, and by Lemma \ref{lem:rkOSrank} this implies 
\[
\mathcal{R}\lb (\ident_2\otimes P)(X)\rb \leq 3.
\]
By \cite[Theorem 3.2.]{cariello2015does} this implies that $(\ident_2\otimes P)(X)\in(\M_2\otimes \M_{d_2})^+$ is separable. This finishes the proof. 
\end{proof}

Using the Choi-Jamiolkowski isomorphism and Lemma \ref{lem:SNtrimmingFrom2EB} the previous theorem implies the following corollary:

\begin{cor}
Any positive bipartite matrix $X\in(\M_{d_1}\otimes \M_{d_2})^+$ with operator Schmidt rank $\mathcal{R}\lb X\rb\leq 3$ satisfies $\text{SN}\lb X\rb\leq \min(d_1,d_2) - 1$.
\label{cor:Cariello}
\end{cor} 

To show that Theorem \ref{thm:2EBfromRank} is non-trivial we will construct an example of a $2$-positive map (actually our example will be a completely positive map) of rank equal to $3$ that is not entanglement breaking. For this consider first the quantum states 
\[
\rho_1 = \frac{1}{6}\begin{pmatrix} 2 & 1 & 0\\ 1 & 2 & 1\\ 0 & 1 & 2\end{pmatrix}, \quad\rho_2 = \frac{1}{6}\begin{pmatrix} 2 & 1 & 0\\ 1 & 2 & -i\\ 0 & i & 2\end{pmatrix}.
\]
Next, we consider the Hermitian matrices
\begin{align*}
H_0 &= \begin{pmatrix} 2.4 & -5.3 & 0\\ -5.3 & 26.7 & 0\\ 0 & 0 & 28.8\end{pmatrix},\\
 H_1 &= \begin{pmatrix} 10.6 & -25+3.2i & 44+33.4i\\ -25-3.2i & 54.6 & -174.4-146.2i\\ 44-33.4i & -174.4+146.2i & 44\end{pmatrix},\\
 H_2 &= \begin{pmatrix} 10.6 & -25-3.2i & -33.4-44i\\ -25+3.2i & 54.6 & 146.2+174.4i\\ -33.4+44i & 146.2-174.4i & 44\end{pmatrix}.
\end{align*}
Finally, we define a linear map $P:\M_3\ra\M_3$ by its Choi matrix
\[
C_P = H_0\otimes \one_3 + \rho_1\otimes H_1 + \rho_2\otimes H_2.
\] 
It can be easily verified that $C_P$ is a positive matrix, and therefore $P$ is a completely positive map. By Lemma \ref{lem:rkOSrank} we have $\text{rk}\lb P\rb\leq 3$, and Theorem \ref{thm:2EBfromRank} implies that $P$ is $2$-entanglement breaking. However, it can be checked that $\vartheta_3\circ P$ is \emph{not} completely positive (since $C_P$ is not PPT). This shows that $P$ is not entanglement breaking.

The previous example comes from \cite[Proposition 6]{heinosaari2012extending} where it is shown that there does not exists any completely positive and trace-preserving map $T:\M_3\ra\M_3$ such that $T(\rho_1)=\rho^T_1$ and $T(\rho_2)=\rho^T_2$. The matrix $C_P$ above is a certificate created by an SDP for this fact. Using the techniques from \cite{heinosaari2012extending} (in particular the SDPs for checking whether a completely positive extension exists) it is possible to construct more examples of this kind.

\subsection{Full entanglement annihilation implies $2$-entanglement breaking}

A completely positive map $T:\M_{d_1}\ra \M_{d_2}$ is called $\infty$-locally entanglement annihilating ($\infty$-LEA) if for any $n\in\N$ and any $X\in (\M_{d_1}^{\otimes n})^+$ there exists a $k\in\N$ and $\lset (Y^i_1,\ldots ,Y^i_n) \rset^k_{i=1}\subset (\M_{d_1}^+)^n$ such that $T^{\otimes n}(X) = \sum^k_{i=1} Y^i_1\otimes Y^i_2\otimes \cdots \otimes Y^i_n$. It is an open problem~\cite{moravvcikova2010entanglement} whether there exist completely positive maps that are $\infty$-LEA but \emph{not} entanglement breaking.  

We need the following lemma:

\begin{lem}[Qubit $\infty$-entanglement annihilating maps]
If the linear map $T:\M_{d_1}\ra \M_{d_2}$ is completely positive and $\infty$-LEA and $d_1=2$ or $d_2=2$, then $T$ is entanglement breaking. 
\label{lem:QbitLEA}
\end{lem}
\begin{proof}
Consider first the case where $d_1=2$ and $d_2=d\geq 2$. If $T$ is not entanglement breaking, then there exists a positive map $P:\M_d\ra\M_2$ such that the composition $P\circ T:\M_2\ra \M_2$ is not completely positive. Now consider the map $Q:\M_2\ra \M_{4}$ given by 
\[
Q = (P\circ T)\otimes \proj{0}{0} + (\vartheta_2\circ P\circ T)\otimes \proj{1}{1},
\]
where $\proj{i}{i}$ for $i=1,2$ denote the computational basis of $\C^2$. It is easy to see, that $Q$ is neither completely positive nor completely copositive. Since $T$ is $\infty$-LEA the map $Q^{\otimes k}$ is positive for all $k\in\N$. This property is called tensor-stable positivity and it has been shown in \cite{muller2016positivity} that any such map has to be completely positive or completely copositive whenever the input or output dimension is $2$. This is a contradiction. 

The case $d_1=d\geq 2$ and $d_2=2$ works analogously by considering a positive map $P:\M_2\ra\M_d$ and reversing the order of $P$ and $T$ in the above proof.

\end{proof}

With the previous lemma we can show the following:

\begin{thm}
If a completely positive map $T:\M_d\ra \M_d$ is $\infty$-LEA, then it is $2$-entanglement breaking.
\end{thm} 

\begin{proof}
Let $T:\M_d\ra \M_d$ be an $\infty$-LEA map. For any $\ket{\psi}\in(\C^2\otimes \C^d)$ we have
\[
(\ident_2\otimes T)(\proj{\psi}{\psi}) = (\ident_2\otimes T\circ \text{Ad}_{V})(\omega_2)
\]
for some $V:\C^2\ra \C^d$ such that $\ket{\psi} = (\one_2\otimes V)\ket{\Omega_2}$ (using Lemma \ref{lem:YAVD}). Now note that the linear map $T\circ \text{Ad}_{V}:\M_2\ra\M_d$ is completely positive and $\infty$-LEA. Therefore, by Lemma \ref{lem:QbitLEA} the map $T\circ \text{Ad}_{V}$ is entanglement breaking, which shows that $(\ident_2\otimes T)(\proj{\psi}{\psi})$ is separable. Since the state $\ket{\psi}$ was chosen arbitrarily this shows that $T$ is $2$-entanglement breaking (see Definition \ref{defn:kEB}). 

\end{proof}

Applying Theorem \ref{thm:SNIteration} we immediately obtain: 

\begin{cor}
If for any $i\in \lset 1,\ldots , d-1\rset$ the linear map $T_i:\M_d\ra \M_d$ is completely positive and $\infty$-LEA, then $T_{d-1}\circ \cdots \circ T_1$ is entanglement breaking.
\end{cor}

\section{Equivalent formulations of the PPT squared conjecture}
\label{sec:PPTsquareEquiv}

\subsection{Composition of different linear maps}

The PPT squared conjecture is formulated in terms of a single linear map that is applied twice. In the following we give a straightforward equivalent conjecture in terms of the composition of two (possibly different) linear maps that are both completely positive and completely copositive. 

\begin{conj}[PPT squared conjecture -- Version 2]
For any pair of linear maps $T_1:\M_{d_1}\ra\M_{d_2}$ and $T_2:\M_{d_2}\ra\M_{d_3}$ that are both completely positive and completely copositive the composition $T_2\circ T_1$ is entanglement breaking.
\label{conj:PPT22}
\end{conj}

For convenience we provide the proof of the equivalence of these two conjectures.

\begin{proof}[Proof of equivalence of Conjecture \ref{conj:PPT21} and Conjecture \ref{conj:PPT22}]\hfill\\
Assume that there are $T_1:\M_{d_1}\ra\M_{d_2}$ and $T_2:\M_{d_2}\ra\M_{d_3}$ that are both completely positive and completely copositive such that the composition $T_2\circ T_1$ is \emph{not} entanglement breaking. Let $d=\max(d_1,d_2,d_3)$ and denote by $V_1:\C^{d_1}\ra \C^{d}$ and $V_2:\C^{d_2}\ra \C^{d}$ the canonical isometries embedding $\C^{d_i}$ for $i=1,2$ into the first coordinates of $\C^d$. Now define the linear maps $\tilde{T}_1:\M_d\ra\M_d$ and $\tilde{T}_2:\M_d\ra\M_d$ as
\begin{align*}
\tilde{T}_1(X) &= T_1(V_1^\dagger XV_1)\oplus 0_{(d-d_2)} \\
\tilde{T}_2(X) &= T_2(V_2^\dagger XV_2)\oplus 0_{(d-d_3)}
\end{align*}
for any $X\in\M_d$. Note that the maps $\tilde{T}_1$ and $\tilde{T}_2$ are both completely positive and completely copositive. Now consider the switch map $T:\M_d\otimes \M_2\ra \M_d\otimes \M_2$ defined as
\[
T(X) = \tilde{T}_1((\one_d\otimes \bra{1}) X (\one_d\otimes \ket{1}))\otimes \ketbra{2} + \tilde{T}_2((\one_d\otimes \bra{2}) X (\one_d\otimes \ket{2}))\otimes \ketbra{1}  
\] 
for any $X\in\M_d\otimes \M_2$. It can be easily verified that the channel $T$ is still completely positive and completely copositive. Now applying this channel twice yields 
\begin{align*}
T\circ T(X)&= \tilde{T}_1\circ \tilde{T}_2 ((\one_d\otimes \bra{2}) X (\one_d\otimes \ket{2})) \otimes \ketbra{2} +\cdots \\
&\cdots + \tilde{T}_2\circ \tilde{T}_1 ((\one_d\otimes \bra{1}) X (\one_d\otimes \ket{1})) \otimes \ketbra{1}
\end{align*}
for any $X\in\M_d\otimes \M_2$. Finally, note that we have
\[
T_2\circ T_1(Y) = T\circ T(V_1YV_1^\dagger\otimes \ketbra{1})
\] 
and by assumption this channel is not entanglement breaking. Therefore, $T\circ T$ cannot be entanglement breaking either. The other direction is clear.

\end{proof}

Note that in the previous proof we increased the dimension of the linear map when using a counterexample for Conjecture \ref{conj:PPT22} to construct a counterexample for Conjecture \ref{conj:PPT21}. We do not know whether this increase of dimension is necessary. 

\subsection{Connection to local entanglement annihilation}

We continue with a reformulation of the PPT squared conjecture related to so called entanglement annihilating channels~\cite{moravvcikova2010entanglement,filippov2012local}. A linear map $T:\M_{d_1}\ra\M_{d_2}$ is called 2-locally entanglement annihilating if the image $\lb T\otimes T\rb(X)$ is separable for any positive matrix $X\geq 0$. Trivial examples of such maps are the entanglement breaking maps. However, there are examples~\cite{filippov2012local} of completely positive maps which are 2-locally entanglement annihilating, but not entanglement breaking. The following reformulation shows that such maps could be obtained from any linear map that is both completely positive and completely copositive:

\begin{conj}[PPT squared conjecture -- Version 3]
For any pair of linear maps $T_1:\M_{d_1}\ra\M_{d_2}$ and $T_2:\M_{d_3}\ra\M_{d_4}$ both of which completely positive and completely copositive the image $(T_1\otimes T_2)(X)$ is separable for any positive matrix $X\in (\M_{d_1}\otimes \M_{d_3})^+$.
\label{conj:PPT23}
\end{conj}

\begin{proof}[Proof of equivalence of Conjecture \ref{conj:PPT22} and Conjecture \ref{conj:PPT23}]\hfill\\
Suppose first that Conjecture \ref{conj:PPT22} holds true. By convexity it suffices to check that $(T_1\otimes T_2)(\ketbra{\psi})$ is separable for any pure state $\ket{\psi}\in\C^{d_1}\otimes C^{d_3}$. By Lemma \ref{lem:YAVD} we can write $\ket{\psi} = (\one_{d_1}\otimes A)\ket{\Omega_{d_1}}$ for a linear transformation $A:\C^{d_1}\ra \C^{d_3}$. Now by Lemma \ref{Lemma:tricks} we have
\[
(T_1\otimes T_2)(\ketbra{\psi}) = \lbr\ident_{d_2}\otimes (T_2\circ \text{Ad}_A\circ \vartheta_{d_1}\circ T_1^*\circ \vartheta_{d_2})\rbr(\ketbra{\Omega_{d_2}}).
\]  
Since both maps $T_2:\M_{d_3}\ra\M_{d_4}$ and $\text{Ad}_A\circ \vartheta_{d_1}\circ T_1^*\circ \vartheta_{d_2}:\M_{d_2}\ra \M_{d_3}$ are completely positive and completely copositive by assumption their composition is entanglement breaking, and by the above equation $(T_1\otimes T_2)(\ketbra{\psi})$ is separable. 

For the other direction assume that Conjecture \ref{conj:PPT23} is true. Given a completely positive and completely copositive map $T:\M_d\ra \M_d$ by assumption we have that  
\[
((\vartheta_d\circ T^*\circ \vartheta_d)\otimes T)(\omega_{d}) = \lb\ident_{d}\otimes T^2\rb(\omega_d)
\]
is separable. This implies directly that $T^2$ is entanglement breaking, and thereby would imply Conjecture \ref{conj:PPT21}. Since Conjecture \ref{conj:PPT21} and Conjecture \ref{conj:PPT22} are equivalent the proof is finished. 

\end{proof}

It should be noted that the previous proof also implies the following equivalence for fixed dimensions: Given a completely positive map $T:\M_{d_2}\ra\M_{d_3}$, the composition $T\circ S$ is entanglement breaking for any completely positive map $S:\M_{d_1}\ra\M_{d_2}$ that is completely copositive if and only if the image $(\tilde{S}\otimes T)(X)$ is separable for any positive matrix $X\in (\M_{d_2}\otimes \M_{d_2})^+$ and any completely positive map $\tilde{S}:\M_{d_2}\ra\M_{d_1}$ that is completely copositive.

\subsection{Decomposability of certain positive maps}

We will begin with a well-known definition: 

\begin{defn}[Decomposable maps]
A positive map $P:\M_{d_1}\ra \M_{d_2}$ is called decomposable if $P=T_1 + \vartheta_{d_2}\circ T_2$ for $T_1,T_2:\M_{d_1}\ra \M_{d_2}$ completely positive
\end{defn}

 We will need the following results on duality of different subcones of the positive maps. This is in the sense of the theory of mapping cones, see \cite{skowronek2009cones} for details.

\begin{thm}[Duality of cones~\cite{skowronek2009cones,horodecki1996separability}]
A linear map $T:\M_{d_1}\ra \M_{d_2}$ is 
\begin{itemize}
\item completely positive and completely copositive if and only if for any decomposable map $P:\M_{d_2}\ra \M_{d_1}$ the composition $P\circ T$ is completely positive.
\item entanglement breaking if and only if for any positive map $P:\M_{d_2}\ra \M_{d_1}$ the composition $P\circ T$ is completely positive.
\end{itemize}
\label{thm:ConeDuality}
\end{thm}

The first point in the previous theorem is due to St\o rmer~\cite{stormer2008duality}  and the second point essentially due to the Horodeckis~\cite{horodecki1996separability}. Using these dualities we can establish another equivalent formulation of the PPT squared conjecture: 

\begin{conj}[PPT squared conjecture -- Version 4]
For any completely positive and completely copositive map $T:\M_{d_1}\ra \M_{d_2}$ and any positive map $P:\M_{d_2}\ra \M_{d_3}$ the composition $P\circ T$ is decomposable.
\label{conj:PPT24}
\end{conj}

\begin{proof}[Proof of equivalence of Conjecture \ref{conj:PPT22} and Conjecture \ref{conj:PPT24}]\hfill\\
Suppose first that Conjecture \ref{conj:PPT22} holds and consider $T:\M_{d_1}\ra \M_{d_2}$ completely positive and completely copositive and $P:\M_{d_2}\ra \M_{d_3}$ positive. For any completely positive and completely copositive map $S:\M_{d_3}\ra \M_{d_1}$ the composition $T\circ S$ is entanglement breaking by assumption. By the second point of Theorem \ref{thm:ConeDuality} the composition $P\circ T\circ S$ is completely positive. Since this holds for any completely positive and completely copositive map $S$ the first point of Theorem \ref{thm:ConeDuality} shows that $P\circ T$ has to be decomposable. 

For the other direction suppose that Conjecture \ref{conj:PPT22} does not hold and let $T:\M_{d_1}\ra \M_{d_2}$ and $S:\M_{d_3}\ra \M_{d_1}$ be a counterexample, i.e. both maps are completely positive and completely copositive but the composition $T\circ S$ is not entanglement breaking. By the second point of Theorem \ref{thm:ConeDuality} there exists a positive map $P:\M_{d_2}\ra \M_{d_3}$ such that $P\circ T\circ S$ is not completely positive. Now by the first point of Theorem \ref{thm:ConeDuality} the composition $P\circ T$ cannot be decomposable.

\end{proof}

Note that with the previous proof we can obtain the following equivalence with fixed dimensions: Given a completely positive map $T:\M_{d_2}\ra\M_{d_3}$, then $T\circ S$ is entanglement breaking for any $S:\M_{d_1}\ra\M_{d_2}$ if and only if $P\circ T$ is decomposable for any positive map $P:\M_{d_3}\ra\M_{d_1}$.  

Using semidefinite programming it is easily checkable whether a given positive map is decomposable. The previous reformulation of the PPT squared conjecture therefore suggests a computational procedure that might lead to a counterexample. First, choose a non-decomposable positive map $P:\M_{d_1}\ra\M_{d_2}$ and use semidefinite programming to find a completely positive and completely copositive map $T:\M_{d_2}\ra \M_{d_1}$ such that $P\circ T:\M_{d_2}\ra\M_{d_2}$ is not completely positive. Then, use semidefinite programming again to check whether the map $P\circ T$ is even non-decomposable. Unfortunately we have not been able to use this procedure to find a counterexample to the PPT squared conjecture.  

Another possibility to construct a counterexample to Conjecture \ref{conj:PPT24} could be to find a tensor-stable positive map, i.e.~a linear map $P:\M_{d_1}\ra\M_{d_2}$ such that $P^{\otimes n}$ is positive for any $n\in\N$, that is neither completely positive nor completely copositive (see \cite{muller2016positivity} for details). It has been shown in \cite{muller2018decomposability} that given such a map there would exist another tensor-stable positive map $\tilde{P}$ and a completely positive map $T$ that is completely copositive such that $(\tilde{P}\circ T)^{\otimes n} = \tilde{P}^{\otimes n}\circ T^{\otimes n}$ is not decomposable for some $n\in\N$. Since $\tilde{P}^{\otimes n}$ is positive and $T^{\otimes n}$ is both completely positive and completely copositive, this would be a counterexample to Conjecture \ref{conj:PPT24}. Unfortunately, it is an open problem whether tensor-stable positive maps exist that are neither completely positive nor completely copositive~\cite{muller2016positivity}.

\section{Examples for the PPT squared conjecture}
\label{sec:ExPPT2}

Here we will collect some examples of completely positive maps $T:\M_d\ra\M_d$ that are completely copositive maps and not entanglement breaking, but for which their composition $T\circ T$ yields entanglement breaking maps. Further examples of this type have been reported in the literature: See for example~\cite{kennedy2017composition} for a family of linear maps based on graphs, and~\cite{collins2018ppt} for a class of random completely positive maps. For these examples the PPT squared conjecture was confirmed.  

\subsection{Completely positive maps with certain symmetries}

Consider a completely positive map $T:\M_{d}\otimes \M_{d}\ra\M_{d}\otimes \M_d$ such that 
\begin{equation}
(\text{Ad}_U\otimes \text{Ad}_V)\circ T\circ (\text{Ad}_{U^T}\otimes \text{Ad}_{V^\dagger}) = T,
\label{equ:SymmCP}
\end{equation}
for all unitaries $U,V\in\mathcal{U}_d$. Completely positive maps satisfying \eqref{equ:SymmCP} correspond (after a suitable reordering of the tensor factors) to Choi matrices with the symmetry
\[
(U\otimes U\otimes V\otimes \overline{V}) C_T (U\otimes U\otimes V\otimes \overline{V})^\dagger = C_T,
\]
for all unitaries $U,V:\mathcal{U}_d$. The set of quantum states with this symmetry has been classified in~\cite{vollbrecht2002activating} and it contains entangled quantum states that have a positive partial transpose (see for example the state called $\overrightarrow{\tau}^{(5)}$ in~\cite[Fig. 1]{vollbrecht2002activating}). However, the corresponding completely positive and completely copositive maps satisfy the PPT squared conjecture:

\begin{prop}
For any pair of completely positive maps $T_1,T_2:\M_{d}\otimes \M_{d}\ra\M_{d}\otimes \M_d$ both completely copositive and satisfying the symmetry \eqref{equ:SymmCP} its composition $T_2\circ T_1$ is entanglement breaking.
\end{prop}

\begin{proof}
Note that for any unitaries $U,V\in\mathcal{U}_d$ we have 
\begin{align*}
T_2\circ T_1 &= (\text{Ad}_U\otimes \text{Ad}_V)\circ T_2\circ (\text{Ad}_{U^T}\otimes \text{Ad}_{V^\dagger})\circ\cdots \\
&\cdots\circ (\text{Ad}_{\overline{U}}\otimes \text{Ad}_V)\circ T_1\circ (\text{Ad}_{U^\dagger}\otimes \text{Ad}_{V^\dagger}) \\
&= (\text{Ad}_U\otimes \text{Ad}_V)\circ T_2\circ T_1\circ (\text{Ad}_{U^\dagger}\otimes \text{Ad}_{V^\dagger}).
\end{align*}
By the Choi-Jamiolkowski isomorphism (and after a suitable reordering of the tensor factors) this shows that 
\[
(U\otimes U\otimes V\otimes V) C_{T_2\circ T_1} (U\otimes U\otimes V\otimes V)^\dagger = C_{T_2\circ T_1}.
\]
It has been shown in \cite[Example 7]{vollbrecht2001entanglement} that positive matrices with this symmetry are separable if and only they have positive partial transpose. This shows that $T_2\circ T_1$ is entanglement breaking.  
\end{proof}

Another example of a completely positive map with unitary symmetries similar to \eqref{equ:SymmCP} that is completely copositive and not entanglement breaking can be obtained from \cite{audenaert2001asymptotic}. For this let $\alpha_d\in (\M_{d}\otimes \M_{d})^+$ denote the normalized projector onto the antisymmetric subspace of $\C^d\otimes \C^d$ given by
\[
\alpha_d := \frac{1}{d(d-1)}\lb \one_d\otimes \one_d - \mathbbm{F}_d\rb,
\]
where $\mathbbm{F}_d\in \M_d\otimes \M_d$ is the flip operator, i.e.~the Choi matrix of the transposition $\mathbbm{F}_d = C_{\vartheta_d}$. Similarly, let $\sigma_d\in (\M_{d}\otimes \M_{d})^+$ denote the normalized projector onto the symmetric subspace of $\C^d\otimes \C^d$ given by 
\[
\sigma_d := \frac{1}{d(d+1)}\lb \one_d\otimes \one_d + \mathbbm{F}_d\rb.
\]
In \cite{audenaert2001asymptotic} the quantum state $\tau\in (\M_{d_A}\otimes \M_{d_B})^{\otimes n}$ for $d=d_A=d_B$ (as a bipartite state with respect to the bipartition into systems labelled $A$ and $B$) given by
\begin{equation}
\tau^n:=\frac{d^n}{d^n+(d+2)^n} \alpha_d^{\otimes n}+ \frac{(d+2)^n}{d^n+(d+2)^n}\lb \frac{1}{d+2} \alpha_d + \frac{d+1}{d+2} \sigma_d\rb^{\otimes n}
\label{equ:Taun}
\end{equation}
arises as the minimizer of the relative entropy distance of $\alpha^{\otimes n}_d$ to the set of states with positive partial transpose (see~\cite{audenaert2001asymptotic} for details). In particular $\tau^n$ has positive partial transpose and for large enough $n\in\N$ it is entangled since otherwise its regularized relative entropy distances to the states with positive partial tranpose and to the separable states would coincide. It has been shown in~\cite{christandl2012entanglement} that this is not the case.

Now, denote by $T_n:\M^{\otimes n}_d\ra\M^{\otimes n}_{d}$ the completely positive map with Choi matrix $C_{T_n} = \tau_n$. By the previous discussion $T_n$ is completely copositive, but not entanglement breaking. However, we have the following theorem:

\begin{prop}
The completely positive map $T_n:\M^{\otimes n}_d\ra\M^{\otimes n}_{d}$ with Choi matrix \eqref{equ:Taun} is completely copositive and not entanglement breaking, but the composition $T_n\circ T_n$ is entanglement breaking.
\end{prop} 

\begin{proof}
Denote by $S:\M_d\ra\M_d$ the completely positive map with Choi matrix $C_S = \sigma_d$, and by $A:\M_d\ra\M_d$ the completely positive map with Choi matrix $C_A = \alpha_d$. It is well known that $\sigma_d$ is separable (see for example \cite[p.29]{horodecki2009quantum}) and therefore $S$ is entanglement breaking. Note that $A = \frac{1}{d(d-1)}W_1$ with the Holevo-Werner map as in \eqref{equ:HWMaps}. Composing with itself gives
\[
A^2(X) = \frac{1}{d^2(d-1)^2}\lb (d-2)\text{Tr}\lbr X\rbr\one_d + X\rb
\]
for any $X\in\M_d$. The Choi matrix of $A^2$ is an (unnormalized) isotropic state and therefore separable since it clearly has positive partial transpose (see for example \cite[p.29]{horodecki2009quantum}). This shows that $A^2$ is entanglement breaking. By \eqref{equ:Taun} we see that 
\[
T_n = \frac{d^n}{d^n+(d+2)^n} A^{\otimes n}+ \frac{(d+2)^n}{d^n+(d+2)^n}\lb \frac{1}{d+2} A + \frac{d+1}{d+2} S\rb^{\otimes n}
\]
and therefore $T^2_n$ will be a convex combination of tensor products of the completely positive maps $A^2, A\circ S, S\circ A, S^2$ all of which are entanglement breaking. Therefore, $T^2_n$ is entanglement breaking as well.

\end{proof}

\subsection{Gaussian channels}
\label{sec:PPTSquGauss}

Before we can define Gaussian channels we have to introduce some formalism. For more details on Gaussian channels and quantum information theory with infinite-dimensional systems see~\cite{holevo2013quantum}. Let $\mathcal{H}=L^2(\R^{n})$ denote the standard $L^2$-space of complex-valued, square-integrable functions on $\R^n$. On this space we can consider two groups of unitary operators $V_x, U_y\in \mathcal{U}(\mathcal{H})$ parametrized by $x,y\in\R^n$ given by
\[
V_x \psi(\xi) = \exp(i\braket{\xi}{x})\psi(\xi)\hspace{0.5cm}\text{ and }\hspace{0.5cm}U_y\psi(\xi) = \psi(\xi + y),
\] 
for $\psi\in \mathcal{H}$ and $x,y,\xi\in \R^n$. Denoting $z=(x_1,y_1,\ldots ,x_n,y_n)\in\R^{2n}$ we can define the system of Weyl unitaries as 
\[
W(z) = \exp(\frac{i}{2}\braket{y}{x})V_xU_y.
\]
These unitaries satisfy the following (Weyl-Segal CCR) relation
\[
W(z)W(z') = \exp(-\frac{i}{2}\sigma_n(z,z'))W(z+z'),
\]
where 
\[
\sigma_n(z,z') = z^T\sigma_n z'
\]
denotes the canonical symplectic form represented by the matrix
\[
\sigma_n = \bigoplus^n_{i=1} \begin{pmatrix} 0 & 1 \\ -1 & 0\end{pmatrix}\in \M(\R^{2n}).
\]
Given a quantum state $\rho\in \mathcal{S}_1\lb\mathcal{H}\rb^+$, i.e. a positive trace-class operator with trace $1$, we can define its characteristic function as
\[
z\mapsto \text{tr}\lb \rho W(z)\rb.
\]
The characteristic function determines the quantum state $\rho$ uniquely. Now we can define the class of Gaussian states as the quantum states $\rho\in \mathcal{S}_1\lb\mathcal{H}\rb^+$ with Gaussian characteristic functions, i.e.~of the form 
\[
z\mapsto \text{tr}\lb \rho W(z)\rb = \exp(i\braket{m}{z} - \frac{1}{2}z^T \gamma z),
\] 
for some $m\in\R^{2n}$ called the mean of $\rho$ and a symmetric matrix $\gamma\in \M\lb\R^{2n}\rb_{\text{sym}}$ called the covariance matrix of $\rho$.  

An important class of quantum channels $T:\mathcal{S}_1(\mathcal{H})\ra \mathcal{S}_1(\mathcal{H})$ are Gaussian channels. These are the trace-preserving completely positive maps preserving the set of Gaussian states. It turns out that every Gaussian channel is uniquely determined on the set of Gaussian states. Since a Gaussian state is uniquely determined by its mean and covariance matrix it is possible to define a Gaussian channel in terms of these quantities as well:

\begin{thm}[Gaussian channels, see~\cite{holevo2013quantum}]
For any Gaussian channel $T:\mathcal{S}_1(\mathcal{H})\ra \mathcal{S}_1(\mathcal{H})$ there exist $X,Y\in \M_{2n}(\R)$ with $Y=Y^T$ and satisfying
\begin{equation}
Y + i(\sigma_n - X\sigma_n X^T)\geq 0,
\label{equ:Gauss}
\end{equation}
such that for any Gaussian state $\rho\in \mathcal{S}_1(\mathcal{H})$ with covariance matrix $\gamma\in \M(\R^{2n})_{\text{sym}}$ the covariance matrix of $T(\rho)$ is given by
\[
\gamma' = X\gamma X^T + Y.
\]
\label{thm:Gauss}
\end{thm}

Many properties of Gaussian channels can be defined in terms of the corresponding matrices $X,Y$ from Theorem \ref{thm:Gauss}. In the following we will focus on these matrices (also called parameters of the channel) and omit the change of the mean of Gaussian states under the action of a Gaussian channel. 

A Gaussian channel $T:\mathcal{S}_1(\mathcal{H})\ra \mathcal{S}_1(\mathcal{H})$ with parameters $X,Y$ is 
\begin{itemize}
\item \textbf{completely copositive} if and only if 
\begin{equation}
Y - i(\sigma_n + X\sigma_n X^T)\geq 0.
\label{equ:GausscoCP}
\end{equation}
\item \textbf{entanglement breaking} if and only if there exist $N,M\in \M_{2n}(\R)$ with $Y=N+M$ satisfying
\begin{equation}
M\geq i\sigma_n \text{ and } N\geq iX\sigma_n X^T.
\label{equ:CondGaussEB}
\end{equation}
\end{itemize}

Proofs for the previous equivalences can be found in~\cite{holevo2013quantum}. Now we can prove the PPT squared conjecture for Gaussian channels: 

\begin{thm}
For any pair of Gaussian channels $T_1,T_2:\mathcal{S}_1(\mathcal{H})\ra \mathcal{S}_1(\mathcal{H})$ each completely copositive, the composition $T_2\circ T_1$ is entanglement breaking. 
\end{thm}
\begin{proof}
For $j\in\lset 1,2\rset$ let $X_j,Y_j\in \M_{2n}(\R)$ with $Y_j=Y_j^T$ denote the parameters of the Gaussian channel $T_j$ (see Theorem \ref{thm:Gauss}). Then the Gaussian channel $T_2\circ T_1$ corresponds to the transformation 
\[
\gamma \mapsto X_2X_1\gamma X_1^TX_2^T + X_2Y_1X_2^T + Y_2
\]
in terms of covariance matrices. To show that this channel is entanglement breaking we choose $N = X_2Y_1X_2^T$ and $M = Y_2$. Adding \eqref{equ:Gauss} and \eqref{equ:GausscoCP} for the channel $T_1$, and multiplying by $X_2$ from the left and $X_2^T$ from the right gives
\[
N = X_2Y_1X_2^T\geq iX_2X_1\sigma_n x_1^TX_2^T.
\] 
Similarly, adding \eqref{equ:GausscoCP} and the transpose of \eqref{equ:Gauss} for the channel $T_2$ gives
\[
M = Y_2\geq i\sigma_n,
\]
where we used that $Y_2=Y^T_2$. By \eqref{equ:CondGaussEB} the previous two inequalities show that the composition $T_2\circ T_1$ is entanglement breaking.

\end{proof}

\section{Conclusion}

We have shown how the Schmidt number can be used to quantify the number of compositions after which certain completely positive maps (or even more general $k$-positive maps) become entanglement breaking. For $n$-entanglement breaking maps, i.e.~maps that break the entanglement with respect to any ancilla system of dimension $n\geq 2$, the Schmidt number iteration technique puts a successively decreasing upper bound on the Schmidt number after each application of such a map to part of a bipartite positive matrix. This leads to an explicit bound on the number of compositions after which an $n$-entanglement breaking map becomes fully entanglement breaking. 

We presented some classes of maps where this technique can be applied, including completely positive maps between matrix algebras of dimension three that are also completely copositive. This proves the PPT squared conjecture in this dimension. Unfortunately, in higher dimensions not all completely positive maps that are completely copositive are even $2$-entanglement breaking. Therefore, further work has to be done to apply our techniques in the same generality as the PPT squared conjecture. A possibility could be to find a fixed number $N\in\N$ (possibly depending on the dimension $d$) such that for every completely positive map $T:\M_d\ra\M_d$ that is also completely copositive, the composition $T^N$ is $2$-entanglement breaking. Using the Schmidt number iteration technique from Section \ref{sec:SNTech} this would imply that the composition $T^{N(d-1)}$ is entanglement breaking. It should be noted that $N=2$ would already follow from a special case of the PPT squared conjecture: If for any completely positive maps $T:\M_d\ra\M_d$ and $S:\M_2\ra\M_d$ both of which completely copositive their composition $T\circ S$ is entanglement breaking, then the composition $T\circ T$ would be $2$-entanglement breaking in general. 

Finally, proving Conjecture \ref{conj:MaxSNPPT} on the Schmidt number of positive matrices with positive partial transposition would imply that for any completely positive map $T:\M_d\ra\M_d$ that is also completely copositive the composition $T^{2^{d-1}-1}$ is entanglement breaking.

\section*{Acknowledgements}
We thank Daniel Cariello for pointing out how his results \cite{cariello2014separability,cariello2015does} together with the techniques from \cite{huber2018high} imply Corollary \ref{cor:Cariello}. MC and AMH acknowledge financial support from the European Research Council (ERC Grant Agreement no 337603) and VILLUM FONDEN via the QMATH Centre of Excellence (Grant No. 10059). MW acknowledges the hospitality of the QMATH Centre.

\appendix

\section{Schmidt number bounds from block structure}
\label{sec:CWDEC}

To make our article selfcontained, we will review here some results introduced in \cite{chen2017schmidt} and \cite{huber2018high} to upper bound the Schmidt number of bipartite quantum states. These results are based on a technique (called Choi decomposition) from \cite{yang2016all} allowing to decompose a $k$-positive map for $k\geq 2$ into the sum of a completely positive map and a $(k-1)$-positive map with reduced input dimension. To present these results, we will need the notion of a trivial lifting of a linear map (see~\cite{yang2016all}).

\begin{defn}[Trivial lifting~\cite{yang2016all}]
 A linear map $L:\M_{d_1}\ra \M_{d_2}$ is called an \emph{$\mathcal{S}$-trivial lifting} for a set $\mathcal{S}\subseteq \lset 1,\ldots ,d_1\rset$ iff $L\lb\proj{i}{j}\rb=0$ whenever $i\in \mathcal{S}$ or $j\in \mathcal{S}$. 
\end{defn}

We will need the so called Choi decomposition from~\cite{yang2016all}.

\begin{thm}[Choi decomposition]
For $k\in \lset 2,\ldots, \min(d_1,d_2)\rset$ and $s\in \lset 1,\ldots ,d_1\rset$, any $k$-positive map $P:\M_{d_1}\ra \M_{d_2}$ can be written as 
\[
P = Q + T\, ,
\]
where $T:\M_{d_1}\ra \M_{d_2}$ is completely positive and $Q:\M_{d_1}\ra \M_{d_2}$ is a $(k-1)$-positive $\lset s\rset$-trivial lifting.
\label{thm:Decomp}
\end{thm}

Note that the Choi decomposion in~\cite{yang2016all} is a slightly different statement with the completely positive map $T$ being non-zero and only giving the existence of an index $s\in \lset 1,\ldots ,d_1\rset$ with the stated properties. However, it is easy to see that the same proof also yields the previous result by allowing the completely positive map $T$ to be zero. As in~\cite{yang2016all} iterating the Choi decomposition yields the following: 

\begin{cor}
For $k\in \lset 2,\ldots, \min(d_1,d_2)\rset$ and any subset $\mathcal{S}\subseteq\lset 1,\ldots ,d_1\rset$ with $\left|\mathcal{S}\right| \leq k-1$, any $k$-positive map $P:\M_{d_1}\ra \M_{d_2}$ can be written as 
\[
P = Q + T\, ,
\]
where $T:\M_{d_1}\ra \M_{d_2}$ is completely positive and $Q:\M_{d_1}\ra \M_{d_2}$ is a $(k-\left|\mathcal{S}\right|)$-positive $\mathcal{S}$-trivial lifting.
\label{cor:Decomp}
\end{cor}

Now we can obtain the following theorem on the structure of positive bipartite matrices with Schmidt number greater than $2$. After completion of this work and of~\cite{huber2018high} we learned that this result can also be found in~\cite[Lemma 15]{chen2017schmidt}. We present a different proof based on the Choi decomposition.

\begin{thm}[Block structure from Schmidt number~\cite{chen2017schmidt}]
For $d_1\leq d_2$ consider a matrix $X\in \lb\M_{d_1}\otimes \M_{d_2}\rb^+$ written as
\[
X = \sum^{d_1}_{i,j=1} \proj{i}{j}\otimes X_{ij}.
\]
with blocks $X_{ij}\in\M_{d_2}$ for $i,j\in\lset 1,\ldots ,d_1\rset$. If $\text{SN}(X)\geq 2$ and $l\in\lset 1,\ldots ,\text{SN}(X)\rset$, then for any 
\[
\lset m_1,\ldots , m_{d_1 - l +2}\rset\subseteq \lset 1,\ldots,d_1\rset,
\]
the principal sub-block matrix 
\[
Y = \sum^{d_1 -l + 2}_{s,t=1} \proj{s}{t}\otimes X_{m_s m_t}\in \lb \M_{d_1 - l +2}\otimes \M_{d_2}\rb^+
\] 
satisfies $\text{SN}(Y)\geq \text{SN}(X)-l+2$. In particular, $Y$ is entangled for $l=\text{SN}(X)$.
\label{thm:SubBlockEntSN}
\end{thm}

\begin{proof}
Setting $\text{SN}\lb X\rb = k$, there exists a $(k-1)$-positive map $P:\M_{d_1}\ra \M_{d_2}$ such that 
\begin{equation}
(P\otimes \mathrm{id}_{d_2})\lb X\rb\ngeq 0\, .
\label{equ:wit}
\end{equation}
For $l\in\lset 1,\ldots ,k\rset$ and distinct indices $m_1,\ldots , m_{d_1 - l +2}\in \lset 1,\ldots,d_1\rset$ we can apply Corollary~\ref{cor:Decomp} for the set $\mathcal{S} = \lset 1,\ldots,d_1\rset\setminus \lset m_1,\ldots , m_{d_1 - l +2}\rset$ with $\left|\mathcal{S}\right| = l-2$ such that
\[
P = Q + T\, ,
\]
where $T:\M_{d_1}\ra \M_{d_2}$ is completely positive and $Q:\M_{d_1}\ra \M_{d_2}$ is a $(k-l+1)$-positive $\mathcal{S}$-trivial lifting. Now we can conclude from \eqref{equ:wit} that 
\[
(Q\otimes \mathrm{id}_{d_2})\lb X\rb = \sum^{d_1}_{i,j=1} Q(\proj{i}{j})\otimes X_{ij} = \sum^{d_1-l+2}_{s,t=1} Q(\proj{m_s}{m_t})\otimes X_{m_sm_t} \ngeq 0\, , 
\]
where we used that $Q(\proj{i}{j})=0$ whenever $i\in \mathcal{S}$ or $j\in \mathcal{S}$. Since the map $Q$ is $(k-l+1)$-positive we have that the positive matrix
\[
Y'=\sum^{d_1-l+2}_{s,t=1} \proj{m_s}{m_t}\otimes X_{m_tm_s} = (V\otimes \one_{d_2})Y (V^\dagger\otimes \one_{d_2})\in \lb \M_{d_1 - l +2}\otimes \M_{d_2}\rb^+
\]
satisfies $\text{SN}(Y')\geq k-l+2$. Here $V:\C^{d_1-l+2}\ra \C^{d_1}$ is the isometry defined by $V\ket{s} = \ket{m_s}$ for $s\in \lset 1,\ldots ,d_1-l+2\rset$, and since the application of a local isometry preserves the Schmidt number the proof is finished.

\end{proof}

The previous theorem allows to convert certain statements about separability of positive matrices $X\in \lb\M_{d_1}\otimes \M_{d_2}\rb^+$ for $d_1=2$ into statements about the Schmidt number for arbitrary $d_1\leq d_2$. For example, it has been shown in \cite{kraus2000separability} that any positive matrix $X\in \lb\M_{2}\otimes \M_{d_2}\rb^+$ invariant under partial transposition on the first ($2$-dimensional) subsystem is separable. Applying the previous theorem in the case of maximal possible Schmidt number (and choosing $l$ maximal) to a general positive matrix invariant under partial transposition on the smaller of its two subsystems immediately implies the following theorem. 

\begin{thm}[Schmidt number of states invariant under partial transposition,~\cite{huber2018high}]
If a positive matrix $X\in \lb\M_{d_1}\otimes \M_{d_2}\rb^+$ with $2\leq d_1\leq d_2$ satisfies $(\vartheta_{d_1}\otimes \ident_{d_2})(X) = X$, then $\text{SN}\lb X\rb\leq d_1 -1$. 
\label{thm:SNPTInv}
\end{thm}

\section{Operator Schmidt rank of Choi matrices}
\label{sec:Appendix2}

Recall that $\mathcal{R}\lb X\rb$ denotes the operator Schmidt rank (see Definition \ref{defn:OSRank}) of the bipartite matrix $X\in\M_{d_1}\otimes \M_{d_2}$. Here we will prove the following lemma, although it is probably well-known.

\begin{lem}
For a linear map $L:\M_{d_1}\ra\M_{d_2}$ we have $\text{rk}\lb L\rb = \mathcal{R}\lb C_L\rb$.
\label{lem:rkOSrank}
\end{lem}

\begin{proof}
Applying the singular value decomposition to $L:\M_{d_1}\ra\M_{d_2}$ as a linear map on the Hilbert-Schmidt inner product space, we find sets of non-zero mutually orthogonal operators $\lset A_i\rset^{\text{rk}\lb L\rb}_{i=1}\subset \M_{d_1}$ and $\lset B_i\rset^{\text{rk}\lb L\rb}_{i=1}\subset\M_{d_2}$ such that 
\[
L(X) = \sum^{\text{rk}\lb L\rb}_{i=1} \text{Tr}\lbr A^\dagger_i X\rbr B_i,
\]
for any $X\in\M_{d_1}$. Now, a simple computation shows that 
\[
C_L = \sum^{\text{rk}\lb L\rb}_{i=1} \overline{A}_i\otimes B_i,
\]
and therefore $\text{rk}\lb L\rb = \mathcal{R}(C_L)$.
\end{proof}

\bibliographystyle{IEEEtran}
\bibliography{wheneb}

\begin{thebibliography}{10}
\providecommand{\url}[1]{#1}
\csname url@samestyle\endcsname
\providecommand{\newblock}{\relax}
\providecommand{\bibinfo}[2]{#2}
\providecommand{\BIBentrySTDinterwordspacing}{\spaceskip=0pt\relax}
\providecommand{\BIBentryALTinterwordstretchfactor}{4}
\providecommand{\BIBentryALTinterwordspacing}{\spaceskip=\fontdimen2\font plus
\BIBentryALTinterwordstretchfactor\fontdimen3\font minus
  \fontdimen4\font\relax}
\providecommand{\BIBforeignlanguage}[2]{{%
\expandafter\ifx\csname l@#1\endcsname\relax
\typeout{** WARNING: IEEEtran.bst: No hyphenation pattern has been}%
\typeout{** loaded for the language `#1'. Using the pattern for}%
\typeout{** the default language instead.}%
\else
\language=\csname l@#1\endcsname
\fi
#2}}
\providecommand{\BIBdecl}{\relax}
\BIBdecl

\bibitem{horodecki2003entanglement}
M.~Horodecki, P.~W. Shor, and M.~B. Ruskai, ``Entanglement breaking channels,''
  \emph{Reviews in Mathematical Physics}, vol.~15, no.~06, pp. 629--641, 2003.

\bibitem{holevo2001evaluating}
A.~S. Holevo and R.~F. Werner, ``{Evaluating capacities of bosonic Gaussian
  channels},'' \emph{Physical Review A}, vol.~63, no.~3, p. 032312, 2001.

\bibitem{horodecki2005secure}
K.~Horodecki, M.~Horodecki, P.~Horodecki, and J.~Oppenheim, ``Secure key from
  bound entanglement,'' \emph{Physical review letters}, vol.~94, no.~16, p.
  160502, 2005.

\bibitem{lami2015entanglement}
L.~Lami and V.~Giovannetti, ``Entanglement--breaking indices,'' \emph{Journal
  of Mathematical Physics}, vol.~56, no.~9, p. 092201, 2015.

\bibitem{de2012quantifying}
A.~De~Pasquale and V.~Giovannetti, ``Quantifying the noise of a quantum channel
  by noise addition,'' \emph{Physical Review A}, vol.~86, no.~5, p. 052302,
  2012.

\bibitem{de2013amendable}
A.~De~Pasquale, A.~Mari, A.~Porzio, and V.~Giovannetti, ``Amendable gaussian
  channels: Restoring entanglement via a unitary filter,'' \emph{Physical
  Review A}, vol.~87, no.~6, p. 062307, 2013.

\bibitem{lami2016entanglement}
L.~Lami and V.~Giovannetti, ``Entanglement-saving channels,'' \emph{Journal of
  Mathematical Physics}, vol.~57, no.~3, p. 032201, 2016.

\bibitem{rahaman2018eventually}
M.~Rahaman, S.~Jaques, and V.~I. Paulsen, ``Eventually entanglement breaking
  maps,'' \emph{Journal of Mathematical Physics}, vol.~59, no.~6, p. 062201,
  2018.

\bibitem{christandl2012PPT}
M.~Christandl, ``{PPT square conjecture},'' \emph{Banff International Research
  Station workshop: Operator structures in Quantum Information Theory}, 2012.

\bibitem{bauml2015limitations}
S.~B{\"a}uml, M.~Christandl, K.~Horodecki, and A.~Winter, ``Limitations on
  quantum key repeaters,'' \emph{Nature communications}, vol.~6, p. 6908, 2015.

\bibitem{christandl2017private}
M.~Christandl and R.~Ferrara, ``Private states, quantum data hiding, and the
  swapping of perfect secrecy,'' \emph{Physical review letters}, vol. 119,
  no.~22, p. 220506, 2017.

\bibitem{kennedy2017composition}
M.~Kennedy, N.~A. Manor, and V.~I. Paulsen, ``{Composition of PPT Maps},''
  \emph{Quantum Information and Computation}, vol.~18, no. 5 \& 6, pp.
  0472--0480, 2018.

\bibitem{terhal2000schmidt}
B.~M. Terhal and P.~Horodecki, ``Schmidt number for density matrices,''
  \emph{Physical Review A}, vol.~61, no.~4, p. 040301, 2000.

\bibitem{choi1975completely}
M.-D. Choi, ``{Completely positive linear maps on complex matrices},''
  \emph{Linear algebra and its applications}, vol.~10, no.~3, pp. 285--290,
  1975.

\bibitem{skowronek2009cones}
{\L}.~Skowronek, E.~St{\o}rmer, and K.~{\.Z}yczkowski, ``Cones of positive maps
  and their duality relations,'' \emph{Journal of Mathematical Physics},
  vol.~50, no.~6, p. 062106, 2009.

\bibitem{chruscinski2006partially}
D.~Chru{\'s}ci{\'n}ski and A.~Kossakowski, ``On partially entanglement breaking
  channels,'' \emph{Open Systems \& Information Dynamics}, vol.~13, no.~1, pp.
  17--26, 2006.

\bibitem{tang1986positive}
W.-S. Tang, ``On positive linear maps between matrix algebras,'' \emph{Linear
  algebra and its applications}, vol.~79, pp. 33--44, 1986.

\bibitem{horodecki1997separability}
P.~Horodecki, ``Separability criterion and inseparable mixed states with
  positive partial transposition,'' \emph{Physics Letters A}, vol. 232, no.~5,
  pp. 333--339, 1997.

\bibitem{sanpera2001schmidt}
A.~Sanpera, D.~Bru{\ss}, and M.~Lewenstein, ``Schmidt-number witnesses and
  bound entanglement,'' \emph{Physical Review A}, vol.~63, no.~5, p. 050301,
  2001.

\bibitem{woronowicz1976positive}
S.~L. Woronowicz, ``{Positive maps of low dimensional matrix algebras},''
  \emph{Reports on Mathematical Physics}, vol.~10, no.~2, pp. 165--183, 1976.

\bibitem{yang2016all}
Y.~Yang, D.~H. Leung, and W.-S. Tang, ``{All 2-positive linear maps from M3(C)
  to M3(C) are decomposable},'' \emph{Linear Algebra and its Applications},
  vol. 503, pp. 233--247, 2016.

\bibitem{huber2018high}
M.~Huber, L.~Lami, C.~Lancien, and A.~M\"uller-Hermes, ``High-dimensional
  entanglement in states with positive partial transposition,'' \emph{Phys.
  Rev. Lett.}, vol. 121, p. 200503, 2018.

\bibitem{gurvits2002largest}
L.~Gurvits and H.~Barnum, ``Largest separable balls around the maximally mixed
  bipartite quantum state,'' \emph{Physical Review A}, vol.~66, no.~6, p.
  062311, 2002.

\bibitem{johnston2013separability}
N.~Johnston, ``Separability from spectrum for qubit-qudit states,''
  \emph{Physical Review A}, vol.~88, no.~6, p. 062330, 2013.

\bibitem{lami2016bipartite}
L.~Lami and M.~Huber, ``Bipartite depolarizing maps,'' \emph{Journal of
  Mathematical Physics}, vol.~57, no.~9, p. 092201, 2016.

\bibitem{cariello2014separability}
D.~Cariello, ``Separability for weakly irreducible matrices,'' \emph{Quantum
  Information \& Computation}, vol.~14, no. 15-16, pp. 1308--1337, 2014.

\bibitem{cariello2015does}
------, ``Does symmetry imply ppt property?'' \emph{Quantum Information \&
  Computation}, vol.~15, no. 9-10, pp. 812--824, 2015.

\bibitem{heinosaari2012extending}
T.~Heinosaari, M.~A. Jivulescu, D.~Reeb, and M.~M. Wolf, ``Extending quantum
  operations,'' \emph{Journal of Mathematical Physics}, vol.~53, no.~10, p.
  102208, 2012.

\bibitem{moravvcikova2010entanglement}
L.~Morav{\v{c}}{\'\i}kov{\'a} and M.~Ziman, ``{Entanglement-annihilating and
  entanglement-breaking channels},'' \emph{Journal of Physics A: Mathematical
  and Theoretical}, vol.~43, no.~27, p. 275306, 2010.

\bibitem{muller2016positivity}
A.~M{\"u}ller-Hermes, D.~Reeb, and M.~M. Wolf, ``{Positivity of linear maps
  under tensor powers},'' \emph{Journal of Mathematical Physics}, vol.~57,
  no.~1, p. 015202, 2016.

\bibitem{filippov2012local}
S.~N. Filippov, T.~Ryb{\'a}r, and M.~Ziman, ``Local two-qubit
  entanglement-annihilating channels,'' \emph{Physical Review A}, vol.~85,
  no.~1, p. 012303, 2012.

\bibitem{horodecki1996separability}
M.~Horodecki, P.~Horodecki, and R.~Horodecki, ``Separability of mixed states:
  necessary and sufficient conditions,'' \emph{Physics Letters A}, vol. 223,
  no.~1, pp. 1--8, 1996.

\bibitem{stormer2008duality}
E.~St{\o}rmer, ``{Duality of cones of positive maps},'' \emph{preprint
  arXiv:0810.4253}, 2008.

\bibitem{muller2018decomposability}
A.~M{\"u}ller-Hermes, ``Decomposability of linear maps under tensor powers,''
  \emph{Journal of Mathematical Physics}, vol.~59, no.~10, p. 102203, 2018.

\bibitem{collins2018ppt}
B.~Collins, Z.~Yin, and P.~Zhong, ``The ppt square conjecture holds generically
  for some classes of independent states,'' \emph{Journal of Physics A:
  Mathematical and Theoretical}, vol.~51, no.~42, p. 425301, 2018.

\bibitem{vollbrecht2002activating}
K.~G.~H. Vollbrecht and M.~M. Wolf, ``Activating distillation with an
  infinitesimal amount of bound entanglement,'' \emph{Physical review letters},
  vol.~88, no.~24, p. 247901, 2002.

\bibitem{vollbrecht2001entanglement}
K.~G.~H. Vollbrecht and R.~F. Werner, ``Entanglement measures under symmetry,''
  \emph{Physical Review A}, vol.~64, no.~6, p. 062307, 2001.

\bibitem{audenaert2001asymptotic}
K.~Audenaert, J.~Eisert, E.~Jan{\'e}, M.~Plenio, S.~Virmani, and B.~De~Moor,
  ``Asymptotic relative entropy of entanglement,'' \emph{Physical review
  letters}, vol.~87, no.~21, p. 217902, 2001.

\bibitem{christandl2012entanglement}
M.~Christandl, N.~Schuch, and A.~Winter, ``Entanglement of the antisymmetric
  state,'' \emph{Communications in Mathematical Physics}, vol. 311, no.~2, pp.
  397--422, 2012.

\bibitem{horodecki2009quantum}
R.~Horodecki, P.~Horodecki, M.~Horodecki, and K.~Horodecki, ``Quantum
  entanglement,'' \emph{Reviews of modern physics}, vol.~81, no.~2, p. 865,
  2009.

\bibitem{holevo2013quantum}
A.~S. Holevo, \emph{Quantum systems, channels, information: a mathematical
  introduction}.\hskip 1em plus 0.5em minus 0.4em\relax Walter de Gruyter,
  2013, vol.~16.

\bibitem{chen2017schmidt}
L.~Chen, Y.~Yang, and W.-S. Tang, ``Schmidt number of bipartite and
  multipartite states under local projections,'' \emph{Quantum Information
  Processing}, vol.~16, no.~3, p.~75, 2017.

\bibitem{kraus2000separability}
B.~Kraus, J.~Cirac, S.~Karnas, and M.~Lewenstein, ``{Separability in 2$\times$
  N composite quantum systems},'' \emph{Physical Review A}, vol.~61, no.~6, p.
  062302, 2000.

\end{thebibliography}

\end{document}